\documentclass[10pt]{article}
\usepackage{xcolor}
\usepackage{amsmath}
\usepackage{amssymb}
\usepackage{amsthm}
\usepackage{graphicx}
\usepackage{dsfont}
\definecolor{labelkey}{rgb}{0,0.08,0.45}
\definecolor{refkey}{rgb}{0,0.6,0.0}
\definecolor{Brown}{rgb}{0.45,0.0,0.05}
\definecolor{dgreen}{rgb}{0.00,0.49,0.00}
\definecolor{dblue}{rgb}{0,0.08,0.75}
\title{
An adjusted payoff-based procedure for normal form games
}
\author{Mario Bravo
\\[5mm]
\small Universidad de Santiago de Chile\\
\small Departamento de Matem\'atica y Ciencia de la Computaci\'on\\
\small {\ttfamily mario.bravo.g@usach.cl}\\[4mm]
}

\newcommand{\escon}[2]{\mathbb E (#1 \, \,  \vert \, \,  #2)}

\newcommand{\RR}{\ensuremath{\mathbb{R}}}
\newcommand{\CC}{\ensuremath{\mathbb{C}}}
\newcommand{\EE}{\ensuremath{\mathbb{E}}}

\newcommand{\NN}{\ensuremath{\mathbb N}}
\newcommand{\ind}{\ensuremath{\mathds 1}}

\newcommand{\PRe}{\ensuremath{\operatorname{Re}}}

\newcommand{\norm}[1]{\left\Vert #1 \right\Vert} 
\newcommand{\deter}[1]{{\operatorname{det}}\left ( #1 \right )}          

\numberwithin{equation}{section}
\newtheorem{theorem}{Theorem}[section]
\newtheorem{definition}[theorem]{Definition}
\newtheorem{proposition}[theorem]{Proposition}
\newtheorem{lemma}[theorem]{Lemma}

\theoremstyle{definition}
\newtheorem{remark}[theorem]{Remark}

\date{}

\begin{document}

\maketitle
\begin{abstract}
We study a simple adaptive model in the framework of an N -player normal form game. The model consists of a repeated game where the players only know their own action space and their own payoff scored at each stage, not those of the other agents. Each player, in order to update her mixed action, computes the average vector payoff she has obtained by using the number of times she has played each pure action. The resulting stochastic process is analyzed via the ODE method from stochastic approximation theory. We are interested in the convergence of the process to rest points of the related continuous dynamics. Results concerning almost sure convergence and convergence with positive probability are obtained and applied to a traffic game. We also provide some examples where convergence occurs with probability zero.
\end{abstract}

{\bfseries MSC2010 Subject classification:}
Primary: 91A26, 91A10 ; Secondary: 62L20, 93E35
\normalsize
\section{Introduction} 
This paper studies an adaptive model for an $N$-player repeated game. We consider boundedly rational players that adapt using simple behavioral rules based on past experience.\\
\indent The decision that a player can make at each stage hinges on the amount of information available. There are several approaches, depending on how much information agents can gather over time. Fictitious play (see Brown \cite{brown51}, Fudenberg and Levine \cite{fl98}) is one of the most widely studied procedures. Players adapt their behavior by performing best responses to the opponent's average past play over time. In this case, each player needs to know her own payoff function and to receive complete information about the other players' moves. A less restrictive framework is when each player is informed of all the possible payoffs she could have obtained by using alternative moves. The exponential procedure (Freund and Shapire \cite{fs99}) is one example of this kind of adaptive process. Note that, in this case, a player does not necessarily observe her payoff function.\\
\indent We are interested in a less informative context here. Players do not anticipate opponents' behavior and we assume that they have no information on the structure of the game. This means that agents have only their own action space and past realized payoffs to react to the environment. We assume that players are given a rule of behavior (a {\em decision rule}) which depends on a state variable. The state variable is updated by a possibly time-dependent rule (an {\em updating rule}) based on the history of play and current observations.\\
\indent A model widely studied in this framework is the {\em cumulative reinforcement learning} procedure, where players conserve a vector perception (the state variable) in which each coordinate of the vector represents how a move performs. The updating rule is defined by adding the payoff received to the component of the previous vector perception corresponding to the move actually played, and keeping the other components unaltered for the unused moves. The decision rule is given by the normalization of this perception vector, assuming that payoffs are positive. Several results for the convergence (and nonconvergence) of players' mixed actions have been obtained (see Beggs~\cite{beggs05}, B\"orgers and Sarin~\cite{bs97}, Laslier et al.~\cite{ltw01}, as well as a normalized version by Posch~\cite{posch97} for the 2-player game framework and Erev and Roth~\cite{er98} for experimental results). In Cominetti et al.~\cite{cms10}, the authors study a model in the same spirit, mainly using a Logit decision rule (which allows nonpositive payoffs) in the $N$-player case. Players update the perception vector by performing an average between the new payoff received and the previous perception. Conditions are given to ensure the convergence to a Nash equilibrium of a perturbed version of the game. A similar model is studied by Leslie and Collins~\cite{lc05}, where results concerning 2-player games are obtained. Another approach using this information framework is proposed by Hart and Mas-Colell~\cite{hmc01}, where the analysis focuses on the convergence of the empirical frequency of play instead of the long-term behavior of the mixed action. Using techniques based on consistent procedures (see Hart and Mas-Colell~\cite{hmc00}), it is shown that, for all games, the set of correlated equilibria is attained.\\
\indent We consider here a particular updating rule where players maintain a perception vector that is updated, on the coordinate corresponding to the action played, by computing the average between the previous perception and the payoff received using the number of times that each action has been played. It is natural to consider this variant: the actions that have been played most often in the past are the ones for which the player should have the most accurate perception, so it is sensible for the player to put less weight on the most recent payoff when updating his perception of this action.\\
\indent The resulting process turns out to be a variation of that explored by Cominetti et al.~\cite{cms10}; but in our case, players use more information on the history of play. Using the tools provided by the stochastic approximation theory (see e.g., Bena\"im~\cite{benaim99}, Benveniste et al.~\cite{bmp90}, Kushner and Yin~\cite{ky03}), the asymptotic behavior of the process can be analyzed by studying a related continuous dynamics. We are interested in the case where players use the Logit decison rule, and our aim is to find general conditions that will lead to almost sure, or with positive probability, convergence to an attractor of the associated ODE. This case is particularly interesting because the rest points of the ODE are the Nash equilibria of a related game.\\
\indent This paper is organized as follows. Section~\ref{sec:aprox_esto} describes the fundamental theory underpinning the stochastic approximation. Section~\ref{sec:model} precisely defines our model in the framework of an infinitely repeated $N$-player normal form game. In Section~\ref{sec:asym}, we restate our algorithm so that it fits the stochastic approximation setting and we provide a general almost sure convergence result. In Section~\ref{sec:logit} we treat the case of the Logit rule in detail. We start by finding an explicit condition to ensure almost sure convergence derived from Section~\ref{sec:asym}. This condition requires the smoothing parameters associated with the Logit rule to be sufficiently small. It is worth noting that, by this point, we have proved that the results obtained for the process studied by Cominetti et al.~\cite{cms10} also hold in our setting. Given this, we compare these two processes in terms of the path-wise rate of convergence. Later, under a weaker assumption, we study convergence to attractors with positive probability. We apply this result to a particular traffic game on a simple network (studied as an application in \cite{cms10}), showing that convergence with positive probability holds under a much weaker assumption than in the general case.  Finally, we provide some examples where the convergence is lost.
\section{Preliminaries}\label{sec:aprox_esto}
This section recalls some basic features of the stochastic approximation theory following the approach in Bena\"im~\cite{benaim99}. The aim is to study the following discrete process in $\RR ^ d$
\begin{equation}
z_{n+1} - z_n = \gamma_{n+1} \big ( H(z_n) + V_{n+1}\big ), \label{general_discreto}
\end{equation}
\noindent where $(\gamma_n)_{n} $ is a nonnegative step-size sequence, $H :\RR^d \to \RR^d $ is a continuous function and $(V_n)_n$ is a (deterministic or random) noise term. Let us denote by $\mathcal L(z_n)$ the limit set of the sequence $(z_n)_n$, i.e., the set of points $z$ such that $\lim_{l \to +\infty}z_{n_l}=z$ for some sequence $n_l \to +\infty$.\\
\indent The connection between the asymptotic behavior of the discrete process (\ref{general_discreto}) and the asymptotic behavior of the continuous dynamics
\begin{equation}
\dot z= H(z)  \label{general_continuo}
\end{equation}
\noindent is obtained as follows. Given $\varepsilon >0$, $T>0$, a set $Z \subseteq \RR ^ d$ and two points $x,y \in Z$, we say that there is an $(\varepsilon, T)$-chain IN $z$ between $x$ and $y$ if there exist $k$ solutions of  (\ref{general_continuo}) $\{\mathbf  x_1, \ldots, \mathbf x_k\}$ and  times $\{ t_1, \ldots, t_k\}$ greater than $T$ such that
\begin{equation*}
\begin{aligned}
&\text{(1) }\mathbf x_i([0,t_i]) \subseteq Z \text{ for all } i \in \{1, \dots, k \},\\
&\text{(2)}\norm{\mathbf x_i (t_i) - \mathbf x_{i+1}(0)} < \varepsilon  \text{ for all } i \in \{1, \dots, k-1 \},\\
&\text{(3)} \norm{\mathbf x_1(0) - x} < \varepsilon \text{ and }  \norm{\mathbf x_k(t_k) - y}< \varepsilon.\\
\end{aligned}
\end{equation*}
\begin{definition}
A set $D \subseteq \RR ^ d$ is {\em Internally Chain Transitive} \textup{(ICT)} for the dynamics \eqref{general_continuo} if it is compact and for all $\varepsilon>0$, $T>0$ and $x,y \in D$ there exists an $(\varepsilon, T)$-chain in $D$ between $x$ and $y$.
\end{definition}

\indent This definition is derived from the notion of  {\em Internally Chain Recurrent} sets introduced by Conley~\cite{conley78}. Roughly speaking, on an ICT set,  we can link any two points by a chain of solutions of the dynamics (\ref{general_continuo}) by allowing small perturbations. ICT sets are compact, invariant and attractor-free. In Bena\"im~\cite{benaim99} the following general theorem is proved.
\begin{theorem}\label{general_teorema}
Consider the discrete process \eqref{general_discreto}. Assume that $H$ is a Lipschitz function and that
\begin{itemize}
\item[$(a)$] the sequence $(\gamma_n)$ is deterministic, $\gamma_n\geq0$, $\sum_n \gamma_n=+ \infty$  and  $\gamma_n \to 0$,
\item[$(b)$] $\sup \limits_{n \in \NN} \norm{z_n} < +\infty$, and
\item[$(c)$] for any $T>0$
\begin{equation*}
\lim \limits _{n \to + \infty }\sup \bigg  \{ \norm{\sum \limits_{i=n}^{k-1}\gamma_{i+1}V_{i+1}}; k \in \{ n+1, \dots, m(\sum \limits_{j=1}^{n} \gamma_j + T)\} \bigg \}=0,
\end{equation*}
\end{itemize}
\noindent where $m(t)$ is the largest integer $l$ such that $t \geq  \sum \limits_{j=1}^{l} \gamma_j $. Then $\mathcal L(z_n) $ is an {\em ICT} set for the dynamics \eqref{general_continuo}.
\end{theorem}
\begin{remark}\label{rb}
In the case where the noise $(V_n)_n$ in \eqref{general_discreto} is a martingale difference sequence with respect to some filtration on a probability space, we say that  \eqref{general_discreto} is a Robbins--Monro~\cite{robmon51} algorithm . In this framework if, for instance, $\sup_n\EE(\norm{V_n}^2)< +\infty$ and $(\gamma_n)_n \in l^2(\NN)$ then assumption ($c$) in Theorem~\ref{general_teorema} holds with probability one (see Bena\"im~\cite[Proposition~4.2]{benaim99}). Moreover this result is still valid if the noise can be decomposed into a martingale difference process plus a random variable that converges almost surely to zero.
\end{remark}
\section{The model}\label{sec:model}
An $N$-player normal form game is introduced as follows. Let $A=\{1,2,\ldots, N \}$ be the set of players. For every $i \in A$ let $S^i$ be the finite action set for player $i$ and let the set $\Delta^i=\{ z \in \RR^{|S^i|}; z^i \geq 0, \sum_i z^i=1\}$ denote her mixed action set. $S=\prod_{i\in A}S^i$ is the set of action profiles and $\Delta=\prod_{i\in A}\Delta^i$ is the set of mixed action profiles. We write as $(s, s^{-i}) \in S$ the action profile where player $i$ uses $s \in S^i$ and her opponents use the profile $s^{-i}\in \prod_{j \neq i}S^j$ and we adopt the same notation when a mixed action profile is involved. The payoff function of each player $i \in A$ is denoted by $G^i:S \to \RR$ and its multilinear extension by $G^i:\Delta \to \RR$.\\
\indent The game is repeated infinitely and we assume that players are not informed about the structure of the game, i.e., neither the number of players (or their strategies) nor the payoff functions are known. At the stage $n \in \NN$, each player $i$ selects an action $s_{n}^i \in S^i$ using the mixed action $\sigma_{n}^i \in \Delta^i$. Then, she obtains her own payoff $g_{n}^i=G^i(s_{n}^{i},s_{n}^{-i})$, and this is the only information she receives.\\
\indent For every $n \in \NN$ and for each player $i$, we assume that the mixed action at stage $n$ , $\sigma_n^i \in \Delta^i$, is determined as a function of a previous {\em perception vector} $x_{n-1}^i \in \RR^{|S^i|}$, i.e., $\sigma_n^i=\sigma^i(x_{n-1}^i)$ with $\sigma^i:\RR^{|S^i|} \to \Delta^i$. The state space for the perception vector profiles $x=(x^1, \dots, x^N) \in \prod_{i \in A}\mathbb R^{|S^i|}$ is denoted $X$. We also assume that, for every $i \in A$, 
\begin{equation}
 \begin{aligned}
 & \text{ the function } \sigma^i:\RR^{|S^i|} \to \Delta^i \text{ is continuous, and }\\
 & \text{ for all } s \in S^i \text{ and }  x^i \in \mathbb R^{|S^i|}, \,\, \sigma^{is}(x^i)>0.
 \end{aligned} \label{hip_sigma} \tag{A}
\end{equation}
\indent We will refer to the function $\sigma: X \to \Delta$ with $\sigma(x)=(\sigma^1(x^1), \ldots,\sigma^N(x^N))$ as the {\em decision rule} of the players.\\
\indent At the end of stage $n$, each player $i$ uses the value $g_{n}^{i}$ and $x_{n-1}^i$ to obtain the new perception vector $x_{n}^i$, and so on. The manner in which $x_n$ is updated is called the {\em updating rule} of the players.\\
\indent Cominetti et al.~\cite{cms10} study the following updating rule
\begin{equation}
 x_{n+1}^{is} =\begin{cases}
\big (1-\gamma_{n+1} \big )x_{n}^{is} + \gamma_{n+1} g_{n+1}^i,  & \text{if } s=s_{n+1}^i, \\
x_{n}^{is}, & \text{otherwise,}
\end{cases} \label{algocms_pre} 
\end{equation}
where we assume that $\gamma_n=\frac{1}{n}$ (see the discussion after Proposition~\ref{rate} for an explanation on this choice).\\
\indent In this paper we consider a variation of \eqref{algocms_pre}. Players will use more information by taking into account the number of times their actions have been played. Explicitly, we define the adjusted process (APD) by
\begin{equation}
x_{n+1}^{is}= \begin{cases}
\displaystyle{\big (1-\frac{1}{\theta_{n+1}^{is}} \big )x_{n}^{is} + \frac{1}{\theta_n^{is}+1}g_{n+1}^i},  & \text{if } s=s_{n+1}^i, \\
x_{n}^{is}, & \text{otherwise,}
\end{cases} \label{algo} \tag{APD} 
\end{equation}
where $\theta_n^{is}$ denotes the number of times action $s$ has been used by player $i \in A$ up to time $n$. Given the particular structure in \eqref{algo}, $x_n$ can be assumed to lie within a compact subset of $X$ for all $n\in \NN$. Note that the new variable is simply an average between the previous one and the new payoff scored. We also notice that \eqref{hip_sigma} implies that the decision rule can be assumed to be component-wise bounded away from zero.\\
\indent As usual, we denote by $\mathcal F_n$ the $\sigma$-algebra generated by the history up to time $n$, $\mathcal F_n=\sigma \big ( (s_m, g_m)_{1 \leq m\leq n}\big )$, where $s_m=(s_m^1, \ldots, s_m^N)$ and $g_m=(g_m^1, \ldots, g_m^N)$.
\section{Asymptotic analysis}\label{sec:asym}
If we want to analyze \eqref{algo} using the tools decribed in Section~\ref{sec:aprox_esto}, the main problem is that we have a stochastic algorithm in discrete time where the step size is random and, moreover, depends on the coordinates of the vector to be updated. Thus, in order to study the asymptotic properties of our adaptive process, let us restate the updating scheme \eqref{algo} in the following manner:
\begin{align*}
 x_{n+1}^{is} - x_{n}^{is} &= \frac{1}{\theta_{n}^{is} + \ind_{\{ s=s_{n+1}^i\}}}(g_{n+1}^i - x_{n}^{is})\ind_{\{ s=s_{n+1}^i\}} \\
&=\frac{1}{(n+1)\lambda_n^{is}}\left [(g_{n+1}^i - x_{n}^{is})\ind_{\{ s=s_{n+1}^i\}}  +b_{n+1}^{is} \right ],
\end{align*}
and $b_{n+1}^{is}=O \left (\frac{1}{n} \right )$, where $\lambda_n^{is}=\frac{\theta_n^{is}}{n}$ is the empirical frequency of action $s$ for player $i$ up to time $n$ and $\ind_C$ stands for the indicator function of set $C$.
\begin{remark}
  Note that the previous decomposition is not well-defined when $\theta_n^{is}=0$, but Lemma~\ref{unacotado} shows that it is  almost surely valid for large $n$ and for all $i$ and $s \in S^i$. 
\end{remark}

 Standard computations involving averages show that
 \begin{equation*}
\lambda_{n+1}^{is} - \lambda_{n}^{is} = \frac{1}{n+1}\big (\ind_{\{ s=s_{n+1}^i\}}-\lambda_n^{is}\big).
\end{equation*}
\indent Then we can express \eqref{algo} differently by introducing the empirical frequency of play. The new form is the (up to a vanishing term) martingale difference scheme
\begin{equation}
 \begin{cases} 
 x_{n+1}^{is} - x_{n}^{is} = \displaystyle{\frac{1}{n+1}\big [ \frac{\sigma^{is}(x_n^i)}{\lambda_n^{is}}(G^i(s,\sigma^{-i}(x_n)) - x_{n}^{is}) + U_{n+1}^{is}\big ]}, \\
\phantom{ }\\
 \lambda_{n+1}^{is} - \lambda_{n}^{is} = \displaystyle{ \frac{1}{n+1}\big [ \sigma^{is}(x_n^i)-\lambda_n^{is} + M_{n+1}^{is}\big ]},
\end{cases}  \label{APD}
\end{equation}
where the noise terms are explicitly
\begin{equation}
\begin{aligned}
U_{n+1}^{is}&= \frac{1}{\lambda_n^{is}}(g_{n+1}^i - x_{n}^{is})\ind_{\{ s=s_{n+1}^i\}} -\bigg [ \frac{\sigma^{is}(x_n^i)}{\lambda_n^{is}}(G^i(s,\sigma^{-i}(x_n)) - x_{n}^{is})\bigg ] + b_{n+1}^{is},\\
&= \frac{1}{\lambda_n^{is}}(g_{n+1}^i - x_{n}^{is})\ind_{\{ s=s_{n+1}^i\}} - \mathbb E \big (\frac{1}{\lambda_n^{is}}(g_{n+1}^i - x_{n}^{is})\ind_{\{ s=s_{n+1}^i\}}\, \vert \,\mathcal F_n \big ) + b_{n+1}^{is},\\
M_{n+1}^{is}&= \ind_{\{ s=s_{n+1}^i\}}-\sigma^{is}(x_n^i), \\
&= \big (\ind_{\{ s=s_{n+1}^i\}}-\lambda_n^{is} \big ) - \escon{\big (\ind_{\{ s=s_{n+1}^i\}}-\lambda_n^{is} \big )}{\mathcal F_n}.\\
\end{aligned}\label{noise}
\end{equation}
\indent From now on, we denote by $\epsilon_n=(U_n,M_n)$ the noise term associated with our process.\\
\indent The scheme \eqref{APD} will allow us to deal with the random (and player-dependent) character of the step size in \eqref{algo}. Now, in the spirit of Theorem~\ref{general_teorema}, the asymptotic behavior of \eqref{APD} is related to the continuous dynamics
\begin{equation}
 \begin{cases}
 \dot{x}_t^{is}= \displaystyle{ \frac{\sigma^{is}(x_t^{i})}{\lambda_t^{is}} \Big(G^i(s,\sigma^{-i}(x_t)) -x_t^{is} \Big )}&=\Psi_x^{is}(x_t,\lambda_t),\\
 \phantom{0}\\
 \dot{\lambda}_t^{is}= \sigma^{is}(x_t^{i})-\lambda_t^{is}&=\Psi_\lambda^{is}(x_t,\lambda_t),
\end{cases} \label{dincont}
\end{equation}
\noindent with $\Psi_x: X \times \Delta \to \prod_{i \in A}\mathbb R^{|S^i|}$ and $\Psi_\lambda: X \times \Delta \to \prod_{i \in A}\Delta_0^i$, and $\Delta_0^i$ standing for the tangent space to $\Delta^i$, i.e., $\Delta_0^i= \{z \in \RR^{|S^i|}; \, \sum_{s \in S^i}z^s=0 \}$. Let us denote $\Psi$ the function defined by $\Psi(x, \lambda)=(\Psi_x(x,\lambda),\Psi_\lambda(x,\lambda))$.\\
\indent For the sake of completeness, let us write the process \eqref{algocms_pre} as 
\begin{equation}
x_{n+1}^{is} - x_{n}^{is} = \frac{1}{n+1} \big [ \sigma^{is}(x_n^i)(G^i(s,\sigma^{-i}(x_n)) - x_{n}^{is}) + \tilde U_{n+1}^{is}\big ], \label{algocms} \\
\end{equation} 
\noindent with the noise term given by
\begin{equation}
\begin{aligned}
\tilde U_{n+1}^{is}&= (g_{n+1}^i - x_{n}^{is})\ind_{\{ s=s_{n+1}^i\}} -\sigma^{is}(x_n^i)(G^i(s,\sigma_n^{-i}(x_n)) - x_{n}^{is}),\\
&= (g_{n+1}^i - x_{n}^{is})\ind_{\{ s=s_{n+1}^i\}} - \escon{(g_{n+1}^i - x_{n}^{is})\ind_{\{ s=s_{n+1}^i\}}}{\mathcal F_n}.
\end{aligned}\label{noisecms}
\end{equation}
Therefore, the corresponding continuous dynamics is given by
\begin{equation}
 \dot{x}_t^{is}=\sigma^{is}(x_t^{i})\big( G^i(s,\sigma^{-i}(x_t)) -x_t^{is} \big )=\Phi^{is}(x_t),\\
 \label{dincontcms}
\end{equation}
\noindent where $\Phi: X \to \prod_{i \in A}\mathbb R^{|S^i|}$.
\begin{remark} \label{dinamicas}
 Observe that the following simple fact holds
\begin{equation*}
 (x, \sigma(x)) \in X \times \Delta \text{ is a rest point of } \eqref{dincont} \Leftrightarrow x \in X \text{  is a rest point of  } \eqref{dincontcms}.
\end{equation*}
\end{remark}
\indent In the following, we will show that asymptotic properties similar to those of \eqref{algocms} can be obtained for our process. This means that explicit conditions can be found to ensure that the process \eqref{APD} converges almost surely to a global attractor for the dynamics \eqref{dincont}.\\
\indent Recall that we have assumed that, for every $n \in \NN$ and $i \in A$, the mixed action $\sigma_n^{i} \in \Delta^i$ is component-wise bounded away from zero. The purpose of the next simple lemma is to verify that the same holds, almost surely, for empirical frequencies of play.
\begin{lemma} \label{unacotado}
For $n \geq 1$, let $\sigma_n$ be a probability distribution over a finite set $T$ and let $i_{n+1}$ be an element of $T$ which is drawn with law $\sigma_{n}$ and assume $(\sigma_n)_n$ is adapted to the natural filtration generated by the history. For all $j \in T$, set 
$$\lambda_n^j=\frac{1}{n}\sum \limits_{p=1}^n \ind_{\{i_p=j\}}.$$ 
Assume that there exists $\overline \sigma>0$ such that $\sigma_n^j\geq\overline \sigma$. Then
$$\liminf \limits_{n \to +\infty} \lambda_n^j \geq \overline \sigma,$$
\noindent almost surely, for every $j \in T$.
\end{lemma}
\begin{proof}
Fix $j \in T$ and let $\mathcal F_{k}$ be the $\sigma$-algebra generated by the history $\{i_1,\ldots, i_k\}$ up to time $k$. Then we have that $\escon{\ind_{\{i_k=j\}}}{\mathcal F_{k-1}}=\sigma_{k-1}^{j} \geq \overline{\sigma}$. On the other hand the random process $(\phi_n^{j})_n$ given by 
$$\phi_n^{j}= \sum \limits_{k=1}^{n} \frac{1}{k} \big (\ind_{\{i_k=j\}} -  \escon{\ind_{\{i_k=j\}}}{\mathcal F_{k-1}}\big )$$
is a martingale and $\sup_{n \in \NN}{(\phi_n^{j})^2}\leq C\cdot \sum_{p\geq 1}\frac{1}{p^2}< +\infty$ for some constant $C$. Hence $(\phi_n^{j})_n$ converges almost surely. Now Kronecker's lemma (see e.g., Shiryaev~\cite[Lemma IV.3.2]{shiryaev96}) gives that
\begin{align*}
\lim \limits_{n \to +\infty} \frac{1}{n}\sum \limits_{k=1}^{n} \big (\ind_{\{i_k=j\}} -  \escon{\ind_{\{i_k=j\}}}{\mathcal F_{k-1}}\big )=0.
\end{align*}
So that $\frac{1}{n}\sum_{k=1}^{n} (\ind_{\{i_k=j\}} -  \escon{\ind_{\{i_k=j\}}}{\mathcal F_{k-1}} ) \leq \lambda_n^{j}  - \overline{\sigma}$. Taking the $\liminf$ we conclude.
\end{proof}
\begin{proposition} \label{conv_ICT}
The process \eqref{APD} converges almost surely to an {\em ICT} set for the continuous dynamics \eqref{dincont}.
\end{proposition}
\begin{proof}
We only have to show that our process satisfies the hypotheses of Theorem~\ref{general_teorema}. The assumptions concerning the regularity of the function involved, the step-size sequence and the boundedness of the process $(x_n,\lambda_n)_n$ hold immediately.\\
\indent According to \eqref{noise}, $M_n$ is almost surely bounded and can be written as a martingale difference scheme plus a vanishing term. Observe that $\escon{U_{n+1}}{\mathcal F_n}=0$ and that $$\norm{U_{n+1}^{is}}\leq C/\lambda_n^{is},$$ for some constant $C$. Then Lemma~\ref{unacotado} implies that $U_n$ is almost surely bounded. In view of Remark~\ref{rb}, assumption $(c)$ of Theorem~\ref{general_teorema} holds for the noise term $\epsilon_n=(U_n, M_n)$  and the conclusion follows.
\end{proof}
\noindent Let us define the function $F: X \to \prod_i \RR^{|S^i|}$  by
\begin{equation}
 F^{is}(x)=G^{i}(s, \sigma^{-i}(x)). \label{f}
\end{equation}
Cominetti et al.~\cite{cms10} show that if the function $F$ is contracting for the infinity norm, then the process \eqref{algocms} converges almost surely to the unique rest point of the dynamics \eqref{dincontcms}. The following result shows that the same holds for the process \eqref{APD} by adding a slighlty stronger assumption on the decision rule $\sigma$. 
\begin{proposition} \label{convergencia}
 Assume that $F$ is contracting for the infinity norm and that, for every $i \in A$, the function $\sigma^i$ is Lipschitz for the infinity norm. Then there exists a unique rest point $(x_*,\sigma(x_*)) \in X \times \Delta$ of \eqref{dincont}. Furthermore, the set $\{(x_*,\sigma(x_*))\}$ is a global attractor and the process \eqref{APD} converges almost surely to $(x_*,\sigma(x_*))$.
\end{proposition}
\begin{proof}
According to Remark~\ref{dinamicas}, $(x_*,\sigma(x_*)) \in X \times \Delta$ is a rest point of \eqref{dincont} if and only if $F(x_*)=x_*$, hence the existence and uniqueness follow from the fact that $F$ is contracting.\\ 
\indent Let $0 \leq L<1$ and $K_i$ be the Lipstchitz constants associated with the functions $F$ and $\sigma^i$, $i \in A$, respectively.  We want to find a suitable strict Lyapunov 
function, i.e., a function $V$ that decreases along the solution paths and that verifies $V^{-1}(\{ 0 \})=\{(x_*,\lambda_*)\}$ with $\lambda_*=\sigma(x_*)$. Let $V:X \times \Delta \to \RR_+$ be defined by 
$$V(x, \lambda)= \max \Big \{ \norm{x - x_{*}}_{\infty},\frac{1}{\zeta}\norm{\lambda - \lambda_{*}}_{\infty}\Big \},$$ 
where $\zeta>0$ will be defined later. Function $V$ is the maximum of a finite number of smooth functions, therefore it is absolutely continuous and its derivatives are the evaluation of the derivatives of the function attaining the maximum. We distinguish two cases:\\
{\sc Case 1.} $V(x_t, \lambda_t)=\norm{x_t - x_{*}}_{\infty}$. Let $i \in A$ and $s \in S^i$  be such that $V(x_t, \lambda_t)=|x_t^{is} - x_{*}^{is}|$. Let us assume that $x_t^{is} - x_{*}^{is} \geq 0$. Then, for almost all $t \in \RR$, 
\begin{align*}
\frac{d}{dt} V (x_t, \lambda_t) & = \frac{d}{dt} (x_t^{is} - x_*^{is})= \frac{\sigma^{is}(x_t^{i})}{\lambda_t^{is}} \big( F^{is}(x_t) - F^{is}(x_*) + x_*^{is}- x_t^{is} \big ),\\
&\leq -\xi(1 - L)\norm{x_t - x_*}_{\infty}=-\xi(1 - L)V (x_t, \lambda_t),
\end{align*}
for some $\xi>0$ such that $\sigma^{is}(x)\geq \xi$ for every $i \in A$ and $s\in S^i$. If $x_t^{is} - x_{*}^{is} < 0$, the computations are analogous.\\ 
{\sc Case 2.} $V(x_t,\lambda_t)=\frac{1}{\zeta}\norm{\lambda_t - x_{*}}_{\infty}$.  Let $i \in A$ and $s \in S^i$  be such that $V(x_t, \lambda_t)=\frac{1}{\zeta}|\lambda_t^{jr}-\lambda_*^{jr}|$. We also assume that $\lambda_t^{jr}-\lambda_*^{jr} \geq 0$. Then, for almost all $t \in \RR$, 
\begin{align*}
\frac{d}{dt} V (x_t, \lambda_t) &=\frac{1}{\zeta}\big [  \sigma^{jr}(x_t^{j})-\sigma^{jr}(x_*)+\lambda_*^{jr}-\lambda_t^{jr} \big ]\\
&\leq - \frac{1}{\zeta}\norm{\lambda_t - \lambda_{*}}_{\infty} + \frac{1}{\zeta}| \sigma^{jr}(x_t^{j})-\sigma^{jr}(x_*) | \\
& \leq -  V (x_t, \lambda_t) + \frac{\max_i K_i}{\zeta} \norm{x_t - x_*}_{\infty}\\
& = -(1-\frac{\max_i K_i}{\zeta})V (x_t, \lambda_t),
\end{align*}
and we take $\zeta>0$ sufficiently large to have $1>\max_i K_i/\zeta$. Again, if the relation $\lambda_t^{jr}-\lambda_*^{jr} < 0$ holds, the computations are the same.\\
\indent Hence $V(x_t, \lambda_t) \leq -K V(x_t, \lambda_t)$ for some $K > 0$. So $V$ decreases exponentially fast along the solution paths of the dynamics and $V(x,\lambda)=0$ if and only if $(x,\lambda)=(x_*,\lambda_*)$. Therefore the set $\{ (x_*,\lambda_*)\}$ is a global attractor which is the unique ICT set for \eqref{dincont} (see \cite[Corollary~5.4]{benaim99}). Proposition~\ref{conv_ICT} finishes the proof.
\end{proof}
\section{Logit rule} \label{logit}
The Logit rule is broadly based on the field of discrete choice models as well as game theory. For instance, a model of learning in games where the logit function is used is given by the logit-response dynamics \cite{blume93, afn10}. In this model, the aim is to study the  stochastic stability states of the induced process along with equilibrium selection issues (see \cite{msha12} for a payoff-based implementation of this dynamics and related models).

\noindent  Explicitly, the decision rule $\sigma: X \to \Delta$ is given by
\begin{equation}
 \sigma^{is}(x^i)=\frac{\exp{(\beta_ix^{is})}}{\sum \limits_{r \in S^i} \exp{(\beta_i x^{ir})}}, \label{logit_rule}
\end{equation}
for every $i \in A$ and $s \in S^i$, where $\beta_i>0$ is called the smoothing parameter for player $i$. According to Remark~\ref{dinamicas}, the following result shows that the rest points of the dynamics \eqref{dincont} are the Nash equilibria for an entropy perturbed version of the original game (see Cominetti et al.~\cite{cms10}).
\begin{lemma}\label{pert_game}
Under the Logit decision rule \eqref{logit_rule}, if $x \in X$ is a rest point of the dynamics \eqref{dincontcms}, then $\sigma(x)$ is a Nash equilibrium of a game where the action set for each player $i$ is $\Delta^i$ and her payoff $\overline {G}^i : \Delta \to \RR$ is given by 
\begin{equation}
 \overline {G}^i(\pi)=\sum_{s \in S^i} \pi^{is}G^i(s,\pi^{-i}) - \frac{1}{\beta_i}\sum_{s \in S^i}\pi^{is}\big (\ln(\pi^{is})-1 \big ).
\end{equation}
\end{lemma}
\subsection{Almost sure convergence}\label{sec:logit}
We want to apply Proposition~\ref{convergencia} within this framework. For that purpose, let us introduce the maximum unilateral deviation payoff that a single player can experience, 
\begin{equation}
\eta=\max \limits_{\substack{i \in A, s \in S^i \\ r_1,r_2 \in \tilde S^{-i}}}| G^i(s,r_1) - G^i(s,r_2)|, \label{eta}
\end{equation}
\noindent where $\tilde S^{-i}= \{ (r_1 , r_2) \in S^{-i}\times S^{-i} ; r_1^k \neq r_2^k \text{ for exactly one } k \}$. Now the following proposition ensures that, if the parameters are sufficiently small, the unique attractor is attained with probability one. From now on, we denote $\alpha=\max_{i \in A}\sum_{j \neq i} \beta_j$.
\begin{proposition} \label{prop_logit}
If $2\eta \alpha <1$, the discrete process \eqref{APD} converges almost surely to the unique rest point $(x_*,\sigma(x_*))$ of the dynamics \eqref{dincont}.
\end{proposition}

\begin{proof}
We know from Cominetti et al.~\cite[Proposition~5]{cms10} that, if $2\eta \alpha<1$, function $F$ (defined in \eqref{f}) is contracting for the infinity norm. Observe also that, for every $i \in A$, function $\sigma^i$ is Lipschitz for the infinity norm, since it is a smooth function defined on a compact set. Therefore, Proposition~\ref{convergencia} applies.
\end{proof}
\subsubsection*{Rate of Convergence} \label{seccion_velocidad}
Up to this point, we were able to reproduce some of the theoretical results of the original model (\ref{algocms}) regarding its almost sure convergence to global attractors. Now, we want to justify the inclusion of a counter to the previous actions in terms of the rate of convergence when both learning processes \eqref{APD} and \eqref{algocms} converge almost surely to $(x_*,\lambda_*)$ and $x_*$, respectively, and step size $\gamma_n=\frac{1}{n}$ is considered. This rate of convergence is closely linked to the largest real part eigenvalue of the Jacobian matrix of the functions $\Psi=(\Psi_x,\Psi_\lambda)$ and $\Phi$ at the respective rest points.\\ 
\indent Let us denote $\rho(\mathcal B)$ the maximum real part of the eigenvalues of a matrix $\mathcal B \in \RR^{k \times k}$ , i.e.,
\begin{equation*}
 \rho(\mathcal B)=\max\{\PRe(\mu_j);\, j=1, \ldots, k,  \text{ where }\mu_j \in \mathbb C \text{ is an eigenvalue of the matrix } \mathcal B  \}.\label{rho}
\end{equation*}
We say that a matrix $\mathcal B$ is stable if $\rho(\mathcal B)<0$.
 \begin{lemma} \label{eigenvalues}
Assume that $2\eta \alpha <1$. Let $(x_*,\lambda_*)$ and $x_*$ be the unique rest points of the dynamics \eqref{dincont} and \eqref{dincontcms}, respectively. Then
\begin{equation}
 -1 \leq \rho(\nabla \Psi(x_*,\lambda_*))<-\frac{1}{2}\leq-\frac{N}{\sum \limits_{k \in A} |S^k|} \leq\rho(\nabla \Phi(x_*))<0. \label{ineq_rho}
\end{equation}
\end{lemma}
\begin{proof} Straightforward computations concerning the function $\Psi$ (see \eqref{dincont}) show that 
\begin{equation*}
\frac{\partial \Psi_\lambda^{is}}{\partial x^{jr}}(x_*, \lambda_*)=0 \quad \text{and} \quad \frac{\partial \Psi_\lambda^{is}}{\partial \lambda^{jr}}(x_*, \lambda_*)=-\ind_{\{is=jr\}}, 
\end{equation*}
\noindent for every $i,j \in A$ and $(s,r) \in S^i \times S^j$. Therefore, the matrix $\nabla \Psi(x_*, \lambda_*)$ looks like 
\begin{equation}
\mathbf \nabla \Psi(x_*, \lambda_*)=
\left ( \begin{array}{cc}
\nabla_x \Psi_x(x_*, \lambda_*) & \phantom{-}0 \\
L&-I
\end{array} \right ), \label{matriz}
\end{equation}
\noindent where $I$ stands for the identity matrix and $\nabla_x \Psi_x(x_*, \lambda_*)$ denotes the Jacobian matrix of $\Psi_x$ with respect to $x$ at $(x_*, \lambda_*)$. Notice that the interesting eigenvalues of this matrix are given by its upper-left block because of the zero block and the identity matrix on the right side in \eqref{matriz}. Observe also that $\frac{\partial \Psi_x^{is}}{\partial x^{is}}(x_*, \lambda_*)=-1$, i.e., matrix $\nabla_x \Psi_x(x_*, \lambda_*)$ has diagonal terms equal to $-1$.\\ 
\indent On the other hand, we know that every eigenvalue of a complex matrix $\mathcal B=(\mathcal B_{pq})$ lies within at least one of the Gershgorin discs $D_p(\mathcal B)= \{ z  \in \CC, |z - \mathcal B_{pp}| \leq R_p \}$ where $R_p=\sum_{q \neq p} |\mathcal B_{pq}|$. Given the specific form of matrix $\nabla_x \Psi(x_*, \lambda_*)$ we can estimate the position of its eigenvalues. So, in our case,
\begin{equation*}
R_{is} = \sum_{\substack{ j\in A,\\ j\neq i} }\sum_{ r \in S^j}  \bigg |  \frac{\partial \Psi_x^{is}}{\partial x^{jr}}(x_*, \lambda_*)\bigg |,       
\end{equation*}
\noindent since $\frac{\partial \Psi_x^{is}}{\partial x^{jr}}(x_*,\lambda_*)=0$ if $i=j$ and $r \neq s$. This follows from the fact that $F^{is}(x)$ (defined in \eqref{f}) is independent of the vector $x^i$. Explicitly, 
\begin{equation*}
\frac{\partial \Psi_x^{is}}{\partial x^{jr}}(x_*, \lambda_*)=\beta_j \sigma_*^{jr}\big [ G^i(s,r,\sigma_*^{-(i,j)})- G^i(s,\sigma_*^{-i}) \big ], 
\end{equation*}
\noindent where 
\begin{equation*}
G^i(s,r,\sigma_*^{-(i,j)})=\sum_{\substack{a \in S^{-i} \\ a^j=r} }G^{i}(s, a)\prod_{\substack{ k \neq i \\ k \neq j} }\sigma_*^{k a^k}, 
\end{equation*}
\noindent for $i \neq j$. So that
\begin{align*}
R_{is} &=  \sum_{\substack{j \in A \\ j \neq i}}\beta_j\sum_{ r \in S^j}   \sigma_*^{jr}\big |  G^i(s,r,\sigma_*^{-(i,j)})- G^i(s,\sigma_*^{-i}) \big | \\
& \leq \eta\alpha.
\end{align*}
Then we have that all the eigenvalues of matrix $\nabla_x \Psi_x(x_*, \lambda_*)$ are contained in the complex disc 
\begin{equation}
\{ z  \in \CC, |z +1| \leq \eta\alpha \} \supseteq \bigcup_{\substack {i \in A \\ s \in S^i }}D_{is}(\nabla_x \Psi_x(x_*, \lambda_*)), \label{jugon}
\end{equation}
 which implies that $\rho(\nabla \Psi(x_*,\lambda_*))<-1/2$.\\
\indent Analogous computations involving function $\Phi$ show that
\begin{equation*}
D_{is}(\nabla \Phi(x_*))\subseteq \{ z  \in \CC, |z +\sigma_*^{is}| \leq \sigma_*^{is} \eta\alpha \},
\end{equation*}
for every $i \in A$ and $s \in S^i$. Since $-\sigma_*^{is}+\sigma_*^{is}\eta \alpha<0$,  then $\rho(\nabla \Phi(x_*))<0$.\\
\indent It is obvious that $-1 \leq \rho(\nabla \Psi(x_*,\lambda_*))$. Inequality $-N/\sum_k |S^k| \leq\rho(\nabla \Phi(x_*))$ follows since the trace of matrix  $\nabla \Phi(x_*)$ is equal to $-N$.
\end{proof}
\begin{remark}
Notice that $1/2=N/\sum_k |S^k|$ if and only if $|S_k|=2$ for all $k \in A$.
\end{remark}

\indent The following reduced version of Chen~\cite[Theorem~3.1.1]{chen02} will be useful.
\begin{theorem} \label{teo_chen}
Consider the discrete process given by \eqref{general_discreto}. Assume that the following hold.
\begin{itemize}
\item[$(a)$] For every $n \in \NN$, $\gamma_n>0$,  $\lim_{n \to +\infty}\gamma_n =0$, $\sum_n \gamma_n = +\infty$ and
$$\lim_{ n \to +\infty} \frac{\gamma_n - \gamma_{n+1}}{\gamma_{n+1}\gamma_n}= \overline \gamma  \geq 0.$$
\item[$(b)$] $z_n \to z_0$ almost surely.
\item[$(c)$] There exists $\delta \in (0,1]$ such that
\begin{itemize}
\item[$(c.1)$] for a path such that $z_n \to z_0$, the noise $V_{n}$ can be decomposed into $V_{n}=V'_{n}+V''_{n}$ where
$$\sum_{n \geq 1} \gamma_n^{1 - \delta} V'_{n+1} < +\infty  \quad \text{   and   } \quad V''_{n}=O(\gamma_n^\delta),$$
\item[$(c.2)$] the function $H$ is locally bounded and is differentiable at $z_0$ such that $H(z)= \overline H (z - z_0) + r(z)$ where $r(z_0)=0$ and $r(z)=o(\norm{z - z_0})$ as $z \to z_0$ and 
\item[$(c.2)$] the matrix $\overline H$ is stable and, furthermore,  $\overline H + \delta \overline \gamma I$ is also stable.
\end{itemize} 
\end{itemize}  
Then, almost surely,
$$\epsilon_n(z_n - z_0)\to 0, \quad \text{as    } \: n \to +\infty,$$
for any $\epsilon_n=o((1/\gamma_n)^{\delta})$.
\end{theorem}

\indent The previous result allows us to show that our algorithm is faster. This means that, under the common hypothesis $2\eta\alpha<1$ (which ensures almost sure convergence for both processes), employing the adjusted process \eqref{APD} will help the players to adapt their behavior faster than with the original process \eqref{algocms}.
\begin{proposition} \label{rate}
Assume that $2\eta\alpha<1$ and let  $(x_*,\lambda_*)\in X \times \Delta$ and $x_* \in X$ be the unique rest points of dynamics \eqref{dincont} and \eqref{dincontcms}, respectively. Then the following estimates hold
\begin{itemize}
\item[$(i)$] for almost all trajectories of \eqref{algocms} 
 \begin{equation*}
\varepsilon_n(x_n - x_*)\to 0, \quad \text{as    } \: n \to +\infty,
\end{equation*}
for every sequence $\varepsilon_n= o(n^{|\rho(\nabla \Phi(x_*))|})$, 
\item[$(ii)$]  for almost all trajectories of  \eqref{APD} 
\begin{equation*}
\varepsilon_n\big ( (x_n,\lambda_n) - (x_*,\lambda_*)\big )\to 0, \quad \text{as    } \: n \to +\infty,
\end{equation*}
for every sequence $\varepsilon_n= o(n^{\frac{1}{2}})$.
\end{itemize}
\end{proposition}
\begin{proof} Recall that $\epsilon_{n}=(U_{n}, M_{n})$ and $\tilde U_n$ are the noise terms associated with \eqref{APD} and \eqref{algocms}, respectively (see \eqref{noise} and \eqref{noisecms}). We observe that, for both processes, hypotheses $(a)$ and $(b)$ in Theorem~\ref{teo_chen} are immediately satisfied, since $\gamma_n=\frac{1}{n}$, (with $\overline \gamma=1$) and since Proposition~\ref{prop_logit} applies. Let us verify that condition $(c)$ holds.
\begin{itemize} 
\item[$(i)$] Fix $\delta \in (0, |\rho(\nabla \Phi(x_*))|) $. The random process $(\tilde U_n)_n$ is almost surely bounded and satisfies that $\escon{\tilde U_{n+1}}{\mathcal F_n}=0$. Therefore, $Z_n=\sum_{k=1}^{n} (1/k)^{1 - \delta} \tilde U_{k+1}$ is a martingale where $\sup_n \norm{Z_n}^2 < \sum_{k=1}^{+\infty} (1/k)^{2(1 - \delta)}<+\infty$, and thus convergent (since $\delta < 1/2$). To conclude, observe that function $\Phi$ is smooth and that matrix $\nabla \Phi(x_*) + \delta I$ is stable.
\item[$(ii)$] Fix $\delta \in (0,1/2)$. We repeat the argument by noting that $\epsilon_n=\tilde \epsilon_n + \tilde b_n$ where $\tilde b_n=O(1/n)$ and $\escon{\tilde \epsilon_{n+1}}{\mathcal F_n}=0$. To finish, we use the fact that matrix $\nabla \Psi(x_*,\lambda_*) + \delta I$ is stable since inequality \eqref{ineq_rho} holds.\\ 
\end{itemize}
\end{proof}
 Two important comments are in order. First, as before, let us call $C_n$ the upper-left block of the matrix $\escon{\epsilon_{n+1}^T\epsilon_{n+1}}{\mathcal F_n}$, where,
 \begin{equation*}
   C_n^{is,jr}=
   \begin{cases}
    0 & \text{ if } i\neq j ,\\
  -\left ( \frac{\sigma^{is}(x_n^i)}{\lambda_n^{is}}(G^i(s,\sigma^{-i}(x_n)) - x_{n}^{is})\right ) \cdot  \left ( \frac{\sigma^{ir}(x_n^i)}{\lambda_n^{ir}}(G^i(r,\sigma^{-i}(x_n)) - x_{n}^{ir})\right ) + O \left ( \frac{1}{n} \right) & \text{ if } i=j, s \neq r,\\
\frac{\sigma_n^{is}}{(\lambda_n^{is})^2} \left [ \escon{(G^i(s, s_{n+1}^{-i}) - x_n^{is})^2}{\mathcal F_n} -\sigma_n^{is}(G^i(s, \sigma_{n}^{-i}) - x_n^{is})^2 +\right]   O \left ( \frac{1}{n} \right)& \text{ otherwise }.
\end{cases}
 \end{equation*}
Given that the vector of probabilities $\sigma_n$ converges,  $C_n$ converges almost surely to a deterministic matrix $C$ (which is diagonal since $C_n^{is,jr} \to 0$ when $i=j$ and $s \neq r$). Moreover, $C$ is positive definite since
\begin{equation*}
    C_n^{is,jr}= \frac{\sigma_n^{is}}{(\lambda_n^{is})^2}\left [  \escon{(G^i(s, s_{n+1}^{-i}))^2}{\mathcal F_n} - \sigma_n^{is}(G^i(s, \sigma_{n}^{-i}))^2  + \sigma_n^{is}(1-\sigma_n^{is})x_n^{is}(x_n^{is} - G^i(s, \sigma_n^{is})) \right ]+  O \left ( \frac{1}{n} \right) ,
\end{equation*}
and therefore
\begin{equation*}
  \begin{aligned}
    C^{is,jr}&= \frac{1}{\sigma_*^{is}}\left [  \sum_{s^{-i} \in S^{-i}}(G^i(s, s^{-i}))^2\sigma_*^{-i}(s^{-i}) - \sigma_*^{is}(G^i(s, \sigma_{*}^{-i}))^2 \right ]\\
& > \frac{1}{\sigma_*^{is}}\left [  \sum_{s^{-i} \in S^{-i}}(G^i(s, s^{-i}))^2\sigma_*^{-i}(s^{-i}) - (G^i(s, \sigma_{*}^{-i}))^2 \right ] \geq 0,
\end{aligned}
\end{equation*}
from the fact that $0<\sigma_*^{is}<1$ and the convexity of $x^2$. Hence we can conclude that  $\escon{\epsilon_{n+1}^T\epsilon_{n+1}}{\mathcal F_n}$ converges to a positive definite deterministic matrix and that $\sqrt{n}((x_n,\lambda_n) - (x_*,\lambda_*))$ converges in distribution to a normal random variable (see e.g.. \cite[Theorem 3.3.2]{chen02}). For \eqref{algocms}, considering the continuous function $C(x)=\escon{\tilde U_{n+1}^T\tilde U_{n+1}}{x_n=x}$ and slightly modifying the proof of \cite[Theorem 2.2.12]{duflo97}, it can be shown that $n^{|\rho(\nabla \Phi(x_*))|}(x_n - x_*)$ converges almost surely to a finite random variable if $0<|\rho(\nabla \Phi(x_*))|<1/2$. For instance, considering the game defined by \eqref{matri_j}, we have that $|\rho(\nabla \Phi(x_*))| \approx 0.3$. Figure~\ref{mono} depicts the results of a numerical experience in this particular example where $2\eta\alpha=0.8$.\\

 Second, observe that a better rate can be achieved for \eqref{algocms} if the step size is given by $\gamma_n=\frac{a}{n}$ for $a > |\rho(\nabla \Psi(x_*))|$. This leads to the rate $o(n^{-\delta})$ for all $\delta \in (0,1/2)$. However it is somewhat unrealistic to assume that the players have this information in advance. Nevertheless, we always have that $|\rho(\nabla \Phi(x_*))| < |\rho(\nabla \Psi(x_*,\lambda_*))|$ and thus the scheme \eqref{APD} can reach at least the same path-wise rate of convergence under the hypotheses of Proposition~\ref{rate} and regardless of the step size considered. 
\begin{equation}
 \begin{pmatrix}
(0,0) & (1,0) & (0,1)\\ 
(0,1) & (0,0) & (1,0)\\
(1,0) & (0,1) & (0,0) 
 \end{pmatrix}  \label{matri_j}
 \end{equation}
\begin{figure}[h]
\centering
\includegraphics[scale=0.7]{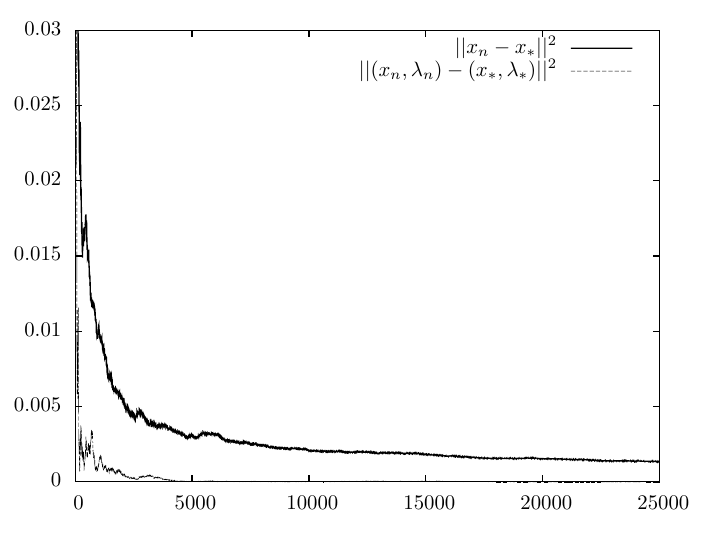}
\caption{$\norm{(x_n,\lambda_n) - (x_*, \lambda_*)}^2$ versus $\norm{x_n - x_*}^2$.\label{mono}}
\end{figure} 

\subsection{Convergence with positive probability}
We use the estimates given by Lemma~\ref{eigenvalues} to extend the range of parameters where general convergence results can be obtained for the process \eqref{APD}. We start by showing that there exists a unique rest point of \eqref{dincont} which is stable if $1 \leq 2\eta \alpha<2$. Let $\mathcal Y \subseteq X \times \Delta$ be the set of rest points of \eqref{dincont} and let $B(\mathcal A)$ be the basin of attraction corresponding to an attractor $\mathcal A$.
\begin{proposition} \label{local}
 Assume that $1 \leq 2\eta \alpha<2$. Then, there exists a unique rest point $(x_*,\lambda_*)$ for the dynamics \eqref{dincont} which is an attractor. 
\end{proposition}
\begin{proof}
Let $(x_*, \lambda_*) \in \mathcal Y$. If $1 \leq 2\eta \alpha<2$, equation \eqref{jugon} shows that matrix $\nabla \Psi(x_*,\lambda_*)$ is stable. To prove that $\{(x_*, \lambda_*)\}$ is an attractor, take $V(x,\lambda)=((x,\lambda) - (x_*, \lambda_*))^T D((x, \lambda) -(x_*, \lambda_*))$ as a (local) Lyapunov function where, for instance, $D$ is  the positive definite solution of the Lyapunov equation $\nabla \Psi(x_*,\lambda_*)^TD + D\nabla \Psi(x_*,\lambda_*)=-I$. Given the fact that basins of attraction cannot overlap, $\mathcal Y$ is finite since $X \times \Delta$ is compact and $\Psi$ is regular. Finally, $\mathcal Y$ reduces to one point since, in this case, it is impossible to have finitely many stable equilibibria due to the Poincar\'e--Hopf Theorem (see e.g., Milnor~\cite[Chapter 6]{milnor97}).
\end{proof}

\indent The following definition is crucial to ensure convergence with positive probability of the process $(x_n,\lambda_n)_n$ to a given (not necessarily global) attractor. 
\begin{definition} Let $(z_n)_n$ be a discrete stochastic process with state space $Z$. A point $z \in Z$ is attainable by  $(z_n)_n$ if for each $m \in \mathbb N$ and every open neighborhood $U$ of $z$, $\mathbb P( \exists n\geq m, \,\, z_n \in U )>0$.
\end{definition}

\indent The following lemma relies on the particular form of the updating rule \eqref{algo} considered in this work.
\begin{lemma} \label{attain}
Fix $\lambda=(\lambda^1, \ldots, \lambda^N) \in \Delta$. Set $x^i \in \RR^{|S^i|}$ such that $x^{is}= G^i(s,\lambda^{-i})$ for all $s \in S^i$ and set $x=(x^1, \ldots, x^N) \in X$. Then, $(x,\lambda) \in X\times \Delta$ is attainable by the process $(x_n, \lambda_n)_n$. In particular, any rest point of \eqref{dincont} is attainable.
\end{lemma}
\begin{proof}
The fact that $\sigma_n^{is} \geq \xi>0$ for every $i \in A$, $s \in S^i$ and $n \in \NN$ implies that any finite sequence generated by \eqref{APD} has positive probability. The updating rule \eqref{algo} can be expressed almost surely, for $n$ sufficiently large, as
\begin{equation}
x_{n+1}^{is}= \frac{1}{\theta_n^{is}}\big (g_{\upsilon^{is}(\theta_n^{is})}^{i} + g_{\upsilon^{is}(\theta_{n}^{is}-1)}^{i}+ \cdots +g_{\upsilon^{is}(1)}^{i} + x_0^{is} \big ) + O \left ( \frac{1}{n}\right ), \label{media}
\end{equation}
where $\upsilon^{is}(k)=\inf \{q\geq 1, \, \theta_q^{is}=k\}$, i.e., the stage when player $i$ has played $s \in S^i$ for the $k$-th time. Observe that we can assume that $m=0$ in the definition of attainability. Let $\zeta_n^{\mathbf s}$ be the number of times that the action profile $\mathbf s \in S$ has been played up to time $n$. Hence, for every $i \in A$ and $s \in S^i$, \eqref{media} implies that
\begin{equation*}
 x_{n+1}^{is}  = \displaystyle{ \sum_{\mathbf r \in S^{-i}}G^i(s, \mathbf r) \frac{\zeta_n^{(s,\mathbf r)}}{\theta_n^{is}} + \overline b_n },
\end{equation*} 
with $\overline b_n=O \big ( (\theta_n^{is})^{-1} \big )$. Observe that $\theta_n^{is} \to +\infty$ almost surely, due to the conditional Borel--Cantelli lemma. Fix $\varepsilon >0$ and let $n$ be an integer such that $k^i_{s}=n \tilde k^i_{s}\in \mathbb N$, where, for every $i \in A$ and $s \in S^i$, $\tilde k^i_{s}$ denotes a rational number satisfying that  $|\lambda^{is}-\tilde k^i_{s}|<\varepsilon$. For a profile $\mathbf s \in S$, let us define the positive integers $n_\mathbf s=\prod_{i \in A}k^i_{\mathbf s^i}$ and $\overline n=\sum_{\mathbf s \in S}n_\mathbf s$. Now we take the sequence generated by \eqref{APD} defined by $l \in \NN$ blocks of size $\overline n$ where within each block, each $\mathbf s \in S$ is played exactly $n_\mathbf s$ times, regardless of the order of play. Fix $i \in A$ and $\mathbf r \in S^{-i}$, so that, by construction
\begin{equation*}
\frac{\zeta_{l \overline n}^{(s,\mathbf r)}}{\theta_{l \overline n}^{is}} = \frac{\prod_{j \neq i }k^j_{\mathbf r^j}k^i_s}{k^i_s\sum \limits_{\mathbf{u} \in S^{-i}}\prod_{j \neq i} k^j_{\mathbf u^j}}= \prod_{j \neq i}\lambda^{j\mathbf r^j} + \tilde b_\varepsilon,
\end{equation*}
\noindent where $\tilde b_\varepsilon \to 0$ as $\varepsilon \to 0$. Finally, given $\varepsilon'>0$, set $l$ large and $\varepsilon$ small to have $\norm{(x_{ln+1}, \lambda_{ln+1})-(x,\lambda)} < \varepsilon'$. 
\end{proof}

\indent Recall that $\mathcal L(z_n)$ is the limit set of sequence $(z_n)_n$. The following result is the goal of this subsection.
\begin{proposition} \label{pos_proba}
If an attractor $\mathcal A$ for the dynamics \eqref{dincont} satisfies that $B(\mathcal A)\cap \mathcal Y \neq \emptyset$, then $\mathbb P(\mathcal L(x_n,\lambda_n)\subseteq \mathcal A)>0$. In particular, under the Logit decision rule \eqref{logit_rule}, if $1 \leq 2\eta \alpha<2$, then $\mathcal Y$ reduces to one point $(x_*, \lambda_*)$ and $\mathbb P((x_n,\lambda_n) \to (x_*,\lambda_*))>0$. 
\end{proposition}

\indent Before providing the proof we need, to briefly introduce the following concepts. Let $\phi$ be the semi-flow induced by the differential equation \eqref{dincont} and $Y_t$ the continuous time affine process associated with the discrete process $(x_n,\lambda_n)_n$, i.e.,
\begin{equation}
 Y(\tau_n +u)=(x_{n},\lambda_n) + u \frac{(x_{n+1},\lambda_{n+1})-(x_n,\lambda_n)}{\tau_{n+1}-\tau_{n}},
\end{equation}
for all $n \in \NN$ and $u \in [0, \frac{1}{n+1})$, where $\tau_n=\sum_{m=1}^{n}\frac{1}{m}$.  Let $(\mathcal F_t)_{t \geq 0}$ be the natural associated filtration.\\
\indent The following technical lemma is now needed. We omit the proof because we keep strictly to the lines of Bena\"im~\cite[Proposition 4.1]{benaim99} along with the explicit computations provided in the proof of Schreiber~\cite[Theorem 2.6]{schreiber01}.
\begin{lemma} \label{tecnico}
For all $T>0$ and $\delta>0$, 
\begin{equation*}
\mathbb P \bigg (\sup \limits_{u \geq t} \big [ \sup \limits_{0 \leq h \leq T}\norm{Y(u +h)-\phi_h(Y(u))} \big ]\geq \delta \,\, \vert \, \, \mathcal F_t \bigg ) \leq \frac{C(\delta,T)}{\exp(ct)}, \end{equation*} 
\noindent for some positive constants $c$ and $C(\delta,T)$ when $t \geq 0$ is large enough.
\end{lemma}
\begin{proof}[Proof of Proposition~\ref{pos_proba}] 
 In view of Proposition~\ref{local} and Lemmas \ref{attain} and \ref{tecnico} the result follows directly from Bena\"im~\cite[Theorem 7.3]{benaim99}.
\end{proof}
Note that for Lemma~\ref{attain} and for the first part of the statement in Proposition~\ref{pos_proba}, we have only assumed condition \eqref{hip_sigma} on decision rule $\sigma$.\\

 The following example shows that the first part of Proposition~\ref{pos_proba} is interesting in its own right. Consider the 2-player zero-sum game defined by the payoff
\begin{equation} \label{zero-sum}
G= \begin{pmatrix}
 0 & -1  \\ 
1 &  \phantom{-}0  \\
 \end{pmatrix}.
\end{equation}  
\noindent Let $(x_*, \sigma(x_*))$, with  $\sigma^1(x_*^1)=\sigma^2(x_*^2)=(1/(1 + e^\beta), e^\beta/(1 + e^\beta))$ and $x_*^1=x_*^2=(-e^\beta/(1 + e^\beta), 1/(1 + e^\beta))$, be the unique rest point of \eqref{dincont}. In this case, every eigenvalue of $\nabla \Psi(x_*, \sigma(x_*))$ is equal to -1. Then $\mathbb P((x_n,\lambda_n) \to (x_*, \sigma(x_*)))>0$ for all $\beta>0$.
\begin{remark}
  Observe that, for any zero-sum game, there exists a unique equilibrium. It is exactly the same proof as in \cite[Theorem 3.2]{hh05}, since if $(x, \lambda)$ is a rest point of \eqref{dincont}, then $\lambda$ is the unique rest point of the perturbed best response dynamics studied.
\end{remark}
\subsubsection*{ A traffic game} 
 The (almost sure or with positive probability) convergence to attractor results obtained when the Logit decision rule is considered are valid under the strong assumption $2\eta\alpha<2$. In fact, this condition becomes very difficult to verify as the number of players increases. Moreover, nonconvergence can occur for some games (see Section~\ref{nonconvergence} for details) if parameter $\eta \alpha$ is large. In this part, we will discuss the interesting application developed in Cominetti et al.~\cite[Section 3]{cms10} and we will show that a result in the spirit of Proposition~\ref{pos_proba} can be obtained under a much weaker condition.\\
\indent Consider a network with a topology that consists of a set of parallel routes. Each route $r \in \mathcal R$ in the network is characterized by an increasing sequence of values $c_1^r \leq \cdots \leq c_N^r$ where $c_u^r$ represents the average travel time when $r$ carries a load of $u$ users. The traffic game is defined as follows. The action set is common to all players, i.e., $S^i= \mathcal R$, for every $i \in A$ with $\mathcal R$ the set of available routes. The payoff to each player $i $, when action profile $\mathbf r \in \mathcal R^N$ is played (i.e., when the network is loaded by the configuration $\mathbf r$), is given by $-c_u^{\mathbf r^i}=G^i(\mathbf r)$, that is, minus her travel time.\\
\indent  This traffic game is shown to be a potential game in the sense that there exists a function $\Lambda: [0,1]^{N \times |\mathcal R|} \to \RR$ such that
\begin{equation*}
 \frac{\partial\Lambda}{\partial \lambda^{is}}(\lambda)= G^{i}(s,\lambda^{-i}),
\end{equation*}
for every $\lambda \in \Delta$. Explicitly, the function $\Lambda$ is given by
\begin{equation}
 \Lambda(\pi)= -\mathbb E_\pi \big [ \sum_{r \in \mathcal R} \sum_{u=1}^{U^r}c_u^r\big ], \label{lambda}
\end{equation}
where the expectation is taken with respect to random variables $U^r=\sum_{i \in A}X^{ir}$ with $X^{ir}$ independent Bernouilli variables such that $\mathbb P(X^{ir}=1)=\pi^{ir}$. It is also shown that the second derivatives of $\Lambda$ are zero except for 
\begin{equation}
\frac{\partial^2 \Lambda}{\partial \pi^{jr} \partial \pi^{ir}}(\pi)=\mathbb E_\pi \big ( c_{U_{ij}^r+1}^r - c_{U_{ij}^r+2}^r \big ) \in [-\eta,0], 
\end{equation}
$i \neq j$, where $U_{ij}^r=\sum_{k\neq i,j}X^{kr}$. Notice that this notion does not correspond to the standard Monderer and Shapley's \cite{ms96} notion of a potential game.\\
\indent We suppose that the smoothing parameters are identical for all players, i.e., $\beta_i=\beta$ for every $i \in A$. Note that, in this framework, $\eta$ (defined in \eqref{eta}) translates to
\begin{equation}
\eta= \max \{ \eta_u^r \, ; \, r \in \mathcal R,\,\, 2 \leq u \leq N \}= \max \{ c_u^r - c_{u-1}^r\, ; \, r \in \mathcal R,\,\, 2\leq u \leq N\}.
\end{equation}
\indent Cominetti et al.~\cite{cms10} obtain the following result.
\begin{proposition} \label{cms_traffic}
If $\eta \beta<1$, then \eqref{dincontcms} has a unique rest point $x_* \in X$ which is symmetric in the sense that $x_*=(\hat x, \ldots, \hat x)$. Furthermore, $\{ x_*\}$ is an attractor for \eqref{dincontcms}.
\end{proposition}

\begin{remark}
The strong requirement (also for the model in \cite{cms10}) on the smoothing parameter in order to ensure uniqueness of equilibrium, can make the prediction of the model very different from the set of Nash equilibria of the stage game. For instance, this is the case in the  two-player congestion game with two links represented by the matrix
\begin{equation*}
 G=\begin{pmatrix}
-2 & -1 \\ 
-1 & -3\\ 
 \end{pmatrix},
 \end{equation*}
where there are two strict equilibria (one player on each route) and one symmetric mixed Nash equilibrium $\hat \sigma=(2/3,1/3)$. In this particular case, we can check numerically that if $\beta \leq \beta^*=0.99$ then there exists a unique equilibrium point $(\sigma_*, x_*)$ of our model. Observe that, naturally, this range is larger than the one derived from our general result $\beta<\eta^{-1}=0.5$. When taking the value $\beta^*$, we have that $\sigma_*^1=\sigma_*^2 =(0.5709,0.4290)$, which is far from the Nash equilibrium $\hat \sigma$. 
\end{remark}

\indent The previous proposition  provides a much weaker condition for the existence and uniqueness of a rest point of \eqref{dincontcms}. Observe also that, despite the fact that the second part yields the existence of an attractor, no convergence result is obtained for the discrete process \eqref{algocms}. The next result shows that, under the assumption $\eta \beta<1$, an additional result can be obtained for \eqref{APD}.

\begin{proposition}
If $\eta \beta<1$, \eqref{dincont} has a unique rest point $(x_*,\lambda_*) \in X \times \Delta$ which is symmetric in the sense that $x_*=(\hat x, \ldots, \hat x)$ and $\lambda_*=(\hat \lambda, \ldots, \hat \lambda)=\sigma(x_*)$. Furthermore,  $\{(x_*,\lambda_*)\}$ is an attractor for \eqref{dincont} and $\mathbb P((x_n, \lambda_n) \to (x_*,\lambda_*))>0$.
\end{proposition}
\begin{proof}
The existence and uniqueness of the symmetric rest point of \eqref{dincont} follows from Remark~\ref{dinamicas} and Proposition~\ref{cms_traffic}. The rest of the proof (below) shows that matrix $\nabla \Psi(x_*,\lambda_*)$ is stable. Hence, $\{(x_*,\lambda_*) \}$ is an attractor for \eqref{dincont} and Proposition~\ref{pos_proba} applies.

Recall that $J_\beta=\nabla_x \Psi_x(x_*,\lambda_*)$ is the upper-left block of matrix $\nabla \Psi(x_*,\lambda_*)$ (see \eqref{matriz}). Observe that, from the definition of $\Psi_x$, the fact that $\sigma^i$ depends only on $x^i$ and \eqref{lambda}, the entries of $J_\beta$ are given by
\begin{align}
J^{is,jr}_\beta&= \sum_{k \in A}\sum_{r' \in \mathcal R}\frac{\partial^2 \Lambda}{\partial \pi^{kr'} \partial \pi^{is}}(\lambda_*)\frac{\partial \sigma^{kr'}}{\partial x^{jr}}(x_*)-\ind_{\{is=jr\}}=\sum_{r' \in \mathcal R}\frac{\partial^2 \Lambda}{\partial \pi^{jr'} \partial \pi^{is}}(\lambda_*)\frac{\partial \sigma^{jr'}}{\partial x^{jr}}(x_*)-\ind_{\{is=jr\}}\nonumber \\
&=\beta \lambda_*^{jr}(1-\lambda_*^{jr}) \mathbb E_{\lambda_*} \big ( c_{U_{ij}^r+1}^r - c_{U_{ij}^r+2}^r \big )\ind_{\{s=r, i \neq j\}}-\ind_{\{is=jr\}}. \label{hessiano}
\end{align}
Since $\lambda_*$ is symmetric ($\lambda^{ir}=\lambda^{jr}$, for all $i,j \in A$), $J_\beta$ is a symmetric matrix. Let us show that $J_\beta$ is negative definite by modifying the trick used in Cominetti et al.~\cite[Proposition 12]{cms10}. Take $h\in \RR^{N|\mathcal R|}\backslash \{0\}$, then, from \eqref{hessiano},
\begin{equation*}
h^T J_\beta h= \sum_{r \in \mathcal R} \big [\beta \sum_{i \neq j}h^{ir}\sqrt{\lambda_*^{ir}(1-\lambda_*^{ir})}h^{jr}\sqrt{\lambda_*^{jr}(1-\lambda_*^{jr})}\mathbb E_{\lambda_*} \big ( c_{U_{ij}^r+1}^r - c_{U_{ij}^r+2}^r \big ) - \sum_{i}(h^{ir})^2 \big ].
\end{equation*}
\indent For every $i \in A$ and $r \in \mathcal R$, put $v^{ir}=h^{ir}\sqrt{\frac{1-\lambda_*^{ir}}{\lambda_*^{ir}}}$, $Z^{ir}=v^{ir}X^{ir}$ and set $\eta_0^r=\eta_1^r=0$. Therefore,  
\begin{align*}
h^T J_\beta h&= \sum_{r \in \mathcal R} \bigg [\beta \sum_{i \neq j}v^{ir}v^{jr}\lambda_*^{ir}\lambda_*^{jr}\mathbb E_{\lambda_*} \big ( c_{U_{ij}^r+1}^r - c_{U_{ij}^r+2}^r \big ) - \sum_{i}\lambda_*^{ir}\frac{(v^{ir})^2}{1-\lambda_*^{ir}} \bigg ] \\
&= \sum_{r \in \mathcal R} \mathbb E_{\lambda_*} \bigg (\beta \sum_{i \neq j}Z^{ir}Z^{jr}( c_{U^r-1}^r - c_{U^r}^r \big ) - \sum_{i}\frac{(Z^{ir})^2}{1-\lambda_*^{ir}} \bigg )\\
&\leq  \sum_{r \in \mathcal R} \mathbb E_{\lambda_*} \bigg (-\eta_{U^r}^r \beta \bigg (\sum_{i}Z^{ir} \bigg )^2 + (\eta_{U^r}^r \beta - 1)\sum_{i}(Z^{ir})^2 \bigg ) <0, 
\end{align*}
\noindent where the last inequality follows by observing that $\eta_{U^r}^r \leq \eta$.
\end{proof}
\subsection{Nonconvergence} \label{nonconvergence}
In order to give an idea of the behavior of the stochastic process defined by \eqref{APD} when $\beta$ (we assume $\beta_i=\beta$ for all $i \in A$) becomes large, we provide a small class of games which underlines the relevance of the hypotheses considered throughout this paper. Consider a 2-player symmetric game, i.e., the action set $\overline S=S^1=S^2$ is common for both players and the payoffs verify that $G^1=(G^2)^T$. Let us assume that $G^1$ has constant-sum by row, which is, $\sum_r G^1(s,r)=k \in \RR$ for every $s \in \overline S$. It is easy to check that for this kind of game there exists a rest point of \eqref{dincont} which has the form $(\overline x, \sigma(\overline x)) \in X \times \Delta$ such that $\overline x^i=(1/k, \ldots, 1/k)$ for $i \in \{1,2\}$. We also assume that $\sum_sG^1(s,s) \neq k$.\\
\indent A game that satisfies the preceding conditions is the {\em good} (resp. {\em bad}) Rock-Scissors-Paper game 
\begin{equation*}\begin{pmatrix}
\phantom{-}0 & \phantom{-}a & -b\\ 
-b & \phantom{-}0 & \phantom{-}a\\
\phantom{-}a & -b & \phantom{-}0
 \end{pmatrix}, \label{RSP} \end{equation*}
\noindent where $ 0<b<a$ (resp. $ 0<a<b$) or the game \eqref{matri_j}.\\
\indent The (strong) hypotheses above ensure that at least one rest point of \eqref{dincont} does not depend on the parameter $\beta$. In the following we will easily show that if $\beta$ is sufficiently large then the rest point $(\overline x, \sigma(\overline x))$ becomes linearly unstable. Later, we will prove that this implies that $\mathbb P((x_n, \lambda_n) \to (\overline x, \sigma(\overline x)))=0$.
\begin{lemma} \label{parte_real}
If $\beta>0$ is sufficiently large, then there exists an eigenvalue $\mu$ of $\nabla\Psi(\overline x, \sigma(\overline x))$ such that $\PRe(\mu)>0$.
\end{lemma}
\begin{proof}
\noindent Again, let $J_{\beta}=\nabla_x \Psi_x(\overline x, \sigma(\overline x))$ be the upper-left block of the Jacobian matrix of $\Psi$, which is the only relevant part, evaluated at $(\overline x, \sigma(\overline x))$. The precise expression for the entries of $J_\beta$ is
\begin{equation*}
J^{is,jr}_\beta=\frac{\partial \Psi_x^{is}}{\partial x^{jr}}(\overline x, \sigma(\overline x))=
\begin{cases} 
-1, \quad & \text{    if }  i=j \text{ and } s = r \\
\phantom{-}0,& \text{    if } i=j \text{ and } s \neq r\\ 
\beta \frac{1}{|\overline S|}\big [ G^i(s,r)- \frac{k}{|\overline S|} \big ], & \text{    otherwise, }\\
\end{cases}
\end{equation*}
\noindent $i,j \in \{ 1,2\}$. Thus $J_\beta$ has the form 
\begin{equation*} \begin{pmatrix}
-I &\phantom{-}\overline J_{\beta}  \\ 
\phantom{-}\overline J_{\beta}& -I  \\
 \end{pmatrix},\end{equation*}  
\noindent with  $\overline J_{\beta} \in \RR^{|\overline S|} \times \RR^{|\overline S|}$. Observe that we can decompose $J_\beta$ as $J_{\beta}=\beta J -I$, where
\begin{equation*} J= \begin{pmatrix}
\mathbf 0 &\overline J  \\ 
\overline J & \mathbf 0  \\
 \end{pmatrix}.\end{equation*}  
\indent Let $\mu_1, \ldots, \mu_{|S|} \in \CC$ be the eigenvalues of $\overline J$ (counting multiplicity). Since we have assumed that $\sum_sG^1(s,s) \neq k$, the trace of $\overline J$ is not zero. Therefore, there exists some eigenvalue $\mu_k \, , k \in \{1, \ldots |\overline S|\}$, with a nonzero real part.  We have that, if $v$ is an eigenvector associated with $\mu_k$, then $\mu_k$ is an eigenvalue of $J$ with corresponding eigenvector $\mathbf u = (v,v) \in \RR^{| \overline S |}\times \RR^{| \overline  S|}$ since
\begin{equation*}
J\mathbf u =\begin{pmatrix}
\mathbf 0 &\overline J  \\ 
\overline J & \mathbf 0  \\
 \end{pmatrix}\begin{pmatrix}v \\ v  \end{pmatrix}= \begin{pmatrix}\overline J v \\ \overline J v  \end{pmatrix}= \mu_k \mathbf u.
 \end{equation*}
\noindent If $\PRe(\mu_k)>0$, the proof is finished. If  $\PRe(\mu_l)\leq 0$ for all $l \in \{1, \ldots |\overline S|\}$ then $\sum_l \PRe(\mu_l) < 0$. Also, the trace of  $J$ is zero and therefore there exists an eigenvalue $\mu$ of $J$ (which is not an eigenvalue of $\overline J$) such that $\PRe(\mu)>0$.\\
\indent Finally, observe that
\begin{equation}
\begin{aligned}
\deter{\beta J -\mu I}= \frac{1}{\beta^{| \overline  S|}}\deter{ J -\frac{\mu}{\beta} I},
\end{aligned} \label{jugo}
\end{equation}
\noindent and it is straightforward from \eqref{jugo} that $\overline \mu$ is an eigenvalue of the matrix $J_\beta$ if $\mu= (1+\overline \mu)/\beta$ is an eigenvalue of $\overline J$. Then $\overline \mu =\beta \mu -1$  whose real part is strictly positive for a sufficiently large $\beta$.
\end{proof}
\begin{proposition}\label{cero_proba}
There exists $\beta>0$ large enough and at least one rest point $(\overline x, \sigma(\overline x)) \in X \times \Delta$  of \eqref{dincont} such that, 
\begin{equation*}
\mathbb P((x_n,\lambda_n) \to (\overline x, \sigma(\overline x)))=0.
\end{equation*}
\end{proposition}
\begin{proof} 
We can directly apply Brandi\`ere and Duflo~\cite[Theorem~1]{bd96}. The hypotheses of the theorem concerning the continuous dynamics and the step size of the discrete process \eqref{APD} are immediately satisfied. The only condition that needs to be verified if how powerfully the noise is projected in a repulsive direction at $(\overline x, \sigma(\overline x))$. Explicitly, it is sufficient to prove that
\begin{equation}
\liminf \limits_{n \to +\infty} \escon{||\epsilon_{n+1}^{pr}||^2}{\mathcal F_n}>0 \text{ a.s. on the event } \Gamma \big (\overline x, \sigma(\overline x)\big )=\{(x_n,\lambda_n) \to (\overline x, \sigma(\overline x)) \}, \label{cond_bd}
\end{equation}
\noindent since the noise term $\epsilon_n=(U_n,M_n)$ is almost surely bounded. Here, the upper-script $pr$ stands for the projection onto the repulsive subspace spanned by the eigenvectors associated with the eigenvalues with a positive real part.\\
\indent Fix $i \in \{1,2\}$, take $\beta$ large to have an eigenvalue $\mu$ of $\nabla \Psi(\overline x, \sigma(\overline x))$ such that $\PRe(\mu)>0$ and let $\mathbf v$ be a correspondent (possibly generalized) eigenvector. The vector $\mathbf v$ has the form $\mathbf v=(v_1 , v_2)$. Note that, necessarily, $v_2 \neq 0$ since, if $v_2=0$, then $v_1$ is a vector of ones, which is indeed an eigenvector for the upper-left block of $\nabla \Psi(\overline x, \sigma(\overline x))$ having -1 as the associated eigenvalue. So that
\begin{align*}
\escon{||\varepsilon_{n+1}^{pr}||^2}{\mathcal F_n}& \geq \escon{\norm{\langle \varepsilon_{n+1}, \mathbf v \rangle \mathbf v}^2}{\mathcal F_n}\geq c\escon{(M_{n+1}^{jr})^2}{\mathcal F_n},
\end{align*}
\noindent with $j=-i$ and for some $r \in \overline S$ and $c>0$. In view of \eqref{noise}, 
\begin{equation*}
\escon{(M_{n+1}^{jr})^2}{\mathcal F_n}= \escon{(\ind_{\{ s_{n+1}^j=r\}} - \sigma^{jr}(x_n^j))^2}{\mathcal F_n} +O \bigg (\frac{1}{n} \bigg ) 
= \sigma^{jr}(x_n^j)(1 - \sigma^{jr}(x_n^j))+O \bigg (\frac{1}{n}\bigg ).
\end{equation*}
\indent To conclude, take the $\liminf_n$ in the previous expression on the event $\Gamma (\overline x, \sigma(\overline x))$ to conclude that \eqref{cond_bd} holds, since $\sigma^{is}$ is bounded away from zero for every $i \in \{1,2 \}$ and $s \in \overline S$.
\end{proof}

\indent As observed by Pemantle~\cite{pemantle90}, nonconvergence results like the previous proposition are not very interesting if the set of unstable points is {\em too large}. The most useful consequences can be stated when this set is finite, as in our example \eqref{matri_j}; moreover, it is easy to check that $(\overline x, \sigma(\overline x))$ is the unique rest point of \eqref{dincont} for all $\beta>0$. The previous result shows that, for a large $\beta$, $(\overline x, \sigma(\overline x))$ has probability zero of being the limit of the process, while for small $\beta$ it is almost surely the limit. More precisely, we have that $\rho(\nabla\Psi(\overline x, \sigma(\overline x)))>0$ if $\beta>3$. Note that, since in this particular case the equilibrium point is known, we can show that $(\overline x, \sigma(\overline x))$ is stable if $2\eta\alpha = 2\beta <6$, i.e. if $\beta<3$. Therefore, by using Proposition~\ref{pos_proba}, we can fully characterize the behavior of the process in this case (except for the case where $\beta=3$). Simulations suggest that there is a cycle that attracts the trajectories and that the empirical frequencies of play still converge to $\sigma(\overline x)$, when $\beta$ is large (see Figure \ref{ciclo}).\\

\vspace{-7ex}
\begin{figure}[hbt!]
  \begin{center}
\includegraphics[scale=0.45]{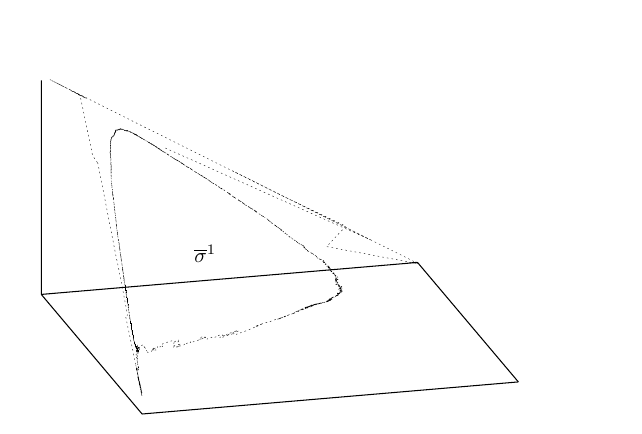}  
    \caption{The mixed action $\sigma_n^1$ of Player 1 when $\beta=4$.}
    \label{ciclo}
  \end{center}
\end{figure}
Finally, note that the same analysis will not work for a general class of games (for instance zero-sum games, as shown by the game given by Equation~\eqref{zero-sum}). Nevertheless, similar analysis can be applied to cases where the game has a unique equilibrium which is known to be unstable. See, for instance, \cite[Chapter 9]{sandholm11}, where this type of study is applied to some of the most well-known dynamics. 
\subsection*{Acknowledgements}
I am deeply indebted to Sylvain Sorin for bringing this problem to my attention and also Michel Bena\"im, Roberto Cominetti, and Mathieu Faure for very helpful discussions and comments.  The development of this project was partially funded by Fondecyt grant No. 3130732, the N\'ucleo Milenio Informaci\'on y Coordinaci\'on en Redes ICM/FIC RC130003 and  by the Complex Engineering Systems Institute (ICM: P-05-004-F, CONICYT: FBO16).
{
 \bibliographystyle{amsplain}
 \bibliography{biblio}
}
\end{document}